\documentclass[runningheads]{llncs}
\usepackage{fullpage}
\usepackage{algorithm}
\usepackage{algcompatible}
\usepackage{amsmath}
\usepackage{amssymb}
\usepackage{comment}
\usepackage{hyperref}
\usepackage{graphics}
\usepackage{graphicx}
\usepackage{xcolor}
\usepackage{minibox}
\usepackage{float}
\hypersetup{
    colorlinks=false,
    citecolor=green,
    urlcolor=cyan
}

\newtheorem{dfn}{Definition}

\makeatletter
\def\ALG@step%
   {%
   \refstepcounter{ALG@line}
   \stepcounter{ALG@rem}
   \ifthenelse{\equal{\arabic{ALG@rem}}{\ALG@numberfreq}}%
      {\setcounter{ALG@rem}{0}\alglinenumber{\arabic{ALG@line}}}%
      {}%
   }%
\makeatother

\date{}
\title{On Some Combinatorial Problems in Cographs}
\author{Kona Harshita \and N. Sadagopan}
\institute{Indian Institute of Information Technology Design and Manufacturing, Kancheepuram, Chennai. \\
\email{\{coe14b016,sadagopan\}@iiitdm.ac.in}
}
\begin{document}
\sloppy
\maketitle
\begin{abstract}
The family of graphs that can be constructed from isolated vertices by disjoint union and graph join operations are called cographs. These graphs can be represented in a tree-like representation termed {\em parse tree or cotree}. In this paper, we study some popular combinatorial problems restricted to cographs. We first present a structural characterization of minimal vertex separators in cographs.  Further, we show that listing all minimal vertex separators and the complexity of some constrained vertex separators are polynomial-time solvable in cographs.  We propose polynomial-time algorithms for connectivity augmentation problems and its variants in cographs, preserving the cograph property. Finally, using the dynamic programming paradigm, we present a generic framework to solve classical optimization problems such as the longest path, the Steiner path and the minimum leaf spanning tree problems restricted to cographs, our framework yields polynomial-time algorithms for all three problems.\\\\
{\bf Keywords: }Cographs, augmentation problems, vertex separators, Hamiltonian path, longest path, Steiner path, minimum leaf spanning tree.
\end{abstract}

\section{Introduction}
\noindent Many scientific problems that arise in practice can be modeled as graph theoretic problems and the solution to which can be obtained through a structural investigation of the underlying graph.  Often, graphs that model scientific problems have a definite structure which inturn help in both structural and algorithmic study.  Special graphs such as bipartite, chordal, planar, cographs etc., have born out of this motivation.  Further, these graphs act as a candidate graph class in understanding the complexity of many classical combinatorial problems, in particular, to understand the gap between NP-complete instances and polynomial-time solvable instances.

It is important to highlight that classical problems such as MIN-VERTEX COVER, MAX-CLIQUE are NP-complete in general graphs, whereas polynomial-time solvable on chordal and cographs.  It is not the case that every NP-complete problem in general graphs is polynomial-time solvable in all special graphs.  For example, the Hamiltonian path, the Steiner tree and the longest path problems remain NP-complete on chordal, planar and $P_5$-free graphs.   For these problems, it is natural to restrict the input further and study the complexity status on subclasses of chordal, planar and $P_5$-free graphs.

The focus of this paper is on cographs, also known as $P_4$-free graphs (graphs that forbid induced $P_4$ ).  Many classical problems such as  STEINER TREE, HAMILTONIAN PATH, LONGEST PATH, MIN-LEAF SPANNING TREE are NP-complete on $P_5$-free graphs.  These results motivated us to look at the complexity status of the above problems in $P_4$-free graphs (cographs).

Cographs are well studied in the literature due its simple structure and it possesses a tree-like representation.   As this tree representation of cographs can be constructed in linear time \cite{corneilrecog}, many classical NP-complete problems have polynomial-time algorithms restricted to cographs.  For instance, HAMILTONIAN PATH (CYCLE) has a polynomial-time algorithm restricted to cographs \cite{corneil}.  Problems such as list coloring, induced subgraph isomorphism and weighted maximum cut remain NP-complete even in cographs. \\
The purpose of this paper is three fold; structural study of cographs from the minimal vertex separator perspective, using these results to present algorithms for listing all minimal vertex separators and to use these results for connectivity augmentation problems and its variants.  We initiate the study of constrained vertex separators in cographs, and show that finding a minimum connected vertex separator and stable vertex separator in cographs are linear-time solvable.

For HAMILTONIAN PATH, LONGEST PATH, STEINER TREE, MIN-LEAF SPANNING TREE, using the {\em parse tree} of cographs, we present polynomial-time algorithms for all of them.  All these problems have a common frame work and make use of the dynamic programming paradigm to obtain an optimum solution.   Our dynamic programming paradigm works with the underlying parse tree, and designing algorithms for graphs by working with the associated tree-like representation has been looked at in \cite{arnborg} for partial $k$-trees.

Given a $k$-vertex (edge) connected graph $G$, the vertex (edge) connectivity augmentation problems ask for a minimum number of edges to be augmented to $G$ so that the resultant graph has the specified vertex (edge) connectivity. This study was initiated by Eswaran et al. in \cite{eshwaran} as it finds applications in the design of robust network design \cite{martin}.

On the complexity front, the $(k+1)$-vertex connectivity augmentation problem of $k$-connected graphs is polynomial-time solvable \cite{vegh}. The edge connectivity augmentation and other related problems are studied in \cite{jordan,nsn,frank,watanabe}.  The algorithm of \cite{vegh} runs in $O(n^7)$ for arbitrary graphs and we present a linear-time algorithm for this problem in cographs.   Connectivity augmentation in special graphs may not preserve the underlying structural properties and hence it is natural to ask for connectivity augmentation algorithms preserving structural properties such as planarity, chordality, $P_4$-freeness.  Towards this end, we shall present a linear-time algorithm for $(k+1)$-vertex connectivity augmentation of $k$-connected graphs in cographs preserving the cograph property.  \\
As far as weighted version of this problem is concerned, it is NP-complete in general graphs \cite{eshwaran,frank}.  We show that weighted version has a polynomial-time algorithm in cographs.   To the best of our knowledge, results presented in this paper do not appear in the literature and we believe that these results convey the message of this paper.

\noindent\textbf{Road map:} In Section \ref{prel}, we shall present the definitions and notation used throughout our work.  We shall present the structural characterization of minimal vertex separators in Section \ref{strucres}. In Section \ref{vcares} and \ref{ecares}, we shall discuss algorithms for connectivity augmentation problems and its variants.  Algorithms for the longest path, the Steiner path and the minimum leaf spanning tree problems are discussed in Section \ref{dpres}.
\section{Preliminaries}
\label{prel}
We shall present graph-theoretic preliminaries first, followed by, definitions and notation related to cographs.

\subsection{Graph-theoretic Preliminaries}
\noindent Throughout our work, we use definitions and notation from \cite{west} and \cite{golumbic}. In this paper, we work with simple, undirected and connected graphs. For a graph $G = (V, E)$, let $V(G)$ denote the vertex set and $E(G)\subseteq \{\{u,v\} ~|~ u,v \in V(G)$ $\}$ denote the edge set.  Let $\overline{G}$ denote the complement of the graph $G$, where $V(\overline{G})=V(G)$ and $E(\overline{G})=\{\{u,v\}~|~\{u,v\}\notin E(G)\}$.  For an edge set $F$, let $G-F$ denote the graph $G=(V,E \setminus F)$ and $G \cup F$ denote the graph $G=(V,E \cup F)$. For $v \in V(G)$, $N_{G}(v) = \{u \in V(G) \mid \{u,v\} \in E(G)\}$ and $\overline{N}_{G}(v) = \{u \in V(G) \mid \{u,v\} \in E(\overline{G})\}$. For $A \subseteq V(G)$ and $v \in V(G)$, let $N_{A}(v)=A \cap N_{G}(v)$. The degree of a vertex $v$ in $G$, denoted as $d_{G}(v)=|N_{G}(v)|$. A graph $H$ is called an induced subgraph of $G$ if for all $u,v \in V(H)$, $\{u,v\} \in E(H)$ if and only if $\{u,v\} \in E(G)$.  For $A \subset V(G)$, let $G[A]$ and $G \setminus A$ denote the induced subgraph of $G$ on vertices in $A$ and $V(G) \setminus A$, respectively. A simple path $P_{uv}$ of a graph $G$ is a sequence of distinct vertices $u=v_1,v_2,\ldots ,v=v_r$ such that $\{v_i,v_{i+1}\} \in E(G), \forall 1 \leq i \leq r-1$ and is denoted by $P_{uv}=(v_1,v_2, \ldots ,v_r)$.  In our work, all paths considered are simple.  Denote a simple path on $n$ vertices by $P_n$.  For a path $P$, let $E(P)$ and $V(P)$ denote the set of edges and vertices, respectively.  For $P_1=(v_1,v_2,\ldots,v_r)$ and $P_2=(w_1,w_2,\ldots,w_s)$ and if $\{v_r,w_1\} \in E(G)$, then $P=(P_1,P_2)$ denote the path $(v_1,v_2,\ldots ,v_r,w_1,\ldots,w_s)$.  A graph $G$ is said to be connected if every pair of vertices in $G$ has a path and if the graph is not connected, it can be divided into disjoint connected components $G_{1},G_{2},\ldots,G_{k}$, $k \geq 2$. A connected component $G_{i}$ is said to be trivial if $|V(G_{i})|=1$ and non-trivial, otherwise. For a connected graph $G$, a subset $S\subset V(G)$ is called a \emph{vertex separator} if $G \setminus S$ is disconnected. A subset $S \subset V(G)$ is called a \emph{minimal vertex separator} if $S$ is a vertex separator and there does not exist a set $S' \subset S$ such that $S'$ is a vertex separator. A subset $S \subset V(G)$ is called a \emph{minimum vertex separator} if it is a minimal vertex separator of least size.   A graph is said to be $k$-connected if there exists a minimum vertex separator of size $k$ in $G$.
\subsection{Cograph Preliminaries}
\noindent We use definitions and notation as in \cite{kirkpatrick,corneil,lerchs}. The graph that can be constructed from isolated vertices by graph join and disjoint union operations recursively is called a cograph.  Also, A graph $G$ is a cograph if every induced subgraph $H$ of $G$ with at least two vertices is either disconnected or the complement to a disconnected graph.  Every cograph can be represented in the form of a binary tree called \emph{parse tree} and is constructed from the operations graph join and disjoint union that are used recursively to construct the cograph.  Each internal node $x$ in the parse tree $T$ is labeled 1 or 0 which indicates the join (1) or union (0) operations in $T$ with respect to the child nodes of $x$.  By construction, parse tree need not be unique.  A unique and normalized form of the parse tree is called \emph{cotree}. For a connected cograph, the root node of the cotree is labelled 1, the children of the node labelled 1 are labelled 0, the children of the node labelled 0 are labelled 1 and so on.  An example is illustrated in Figure \ref{parsetree}.  The root node of $T$ is denoted by $R$.  From the construction of $T$, it can be observed that the set of leaf nodes in $T$ is precisely $V(G)$.  For a node $v\in T$, $N_T(v)=\{w_1,\ldots,w_t\}$, let $G_1,G_2,\ldots,G_t$ denote the subgraphs induced by the leaves in the subtrees rooted at $w_i$ in $T$. If $v$ is labelled 1, then for all $i$, every vertex in $G_i$ is adjacent to every vertex in $G \setminus V(G_i)$ and if $v$ is labelled 0, then no vertex in $G_i$ is adjacent to any vertex in $G\setminus V(G_i)$. For $A \subseteq V(G)$, let $T[A]$ denote the cotree constructed from the cograph $G[A]$.
\begin{figure}[ht]
    \begin{center}
    \includegraphics[height=6cm,width=14cm]{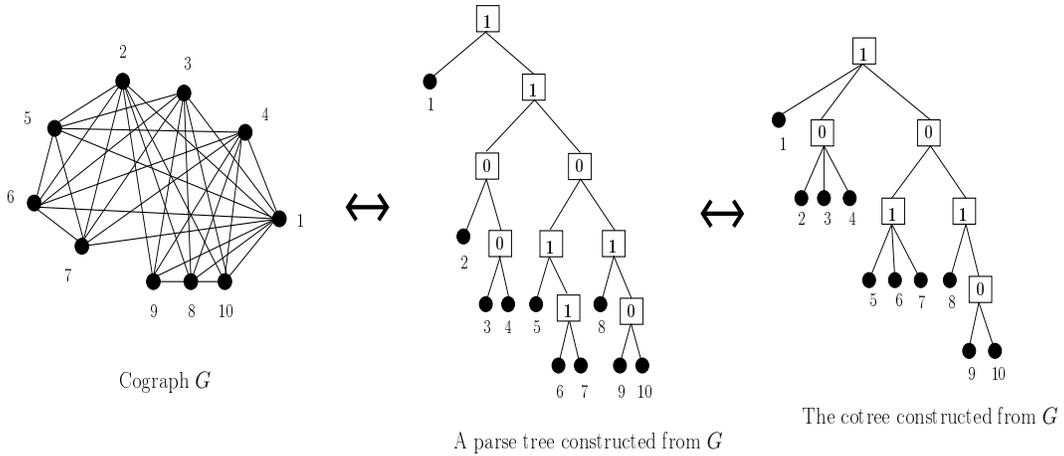}
    \caption{A tree-like representation of a cograph}\label{parsetree}
    \end{center}
\end{figure}
\section{Results on Vertex Separators}
\label{strucres}
In this section, we shall present some structural results with respect to minimal vertex separators in cographs.  It is known from \cite{andras} that a graph $G$ is called a \emph{cograph} if and only if $G$ is \emph{$P_4$-free} (forbids an induced path of length of four).
Using cotree representations of cographs, we shall present an algorithm for listing all minimal vertex separators in cographs and our algorithm runs in linear time. Subsequently, we shall also discuss algorithms for constrained vertex separators restricted to cographs.
\begin{lemma}
\label{lem1}
Let $G$ be a $k$-connected cograph and $S$ be the $k$-size minimal vertex separator of $G$ such that $G\setminus S$ has $G_1, G_2, \ldots, G_k, k \geq 2$ connected components. Then, for every edge $\{u,v\}$ in a non-trivial component, $N_{G}(u) \cap S = N_{G}(v) \cap S$.
\end{lemma}
\begin{proof}
Suppose $G_1$ is a non-trivial component in $G \setminus S$ and $\{u,v\} \in E(G_1)$.  If, on the contrary, there exists a vertex $x \in S$ such that $\{v,x\} \in E(G)$ and $\{u,x\} \notin E(G)$.   Let $y$ be a vertex in $G_2$.  Clearly, the path $(u,v,x,y)$ is an induced path of length 4, contradicting the definition of cographs.  Hence, the claim follows.  \hfill\(\qed\)
\end{proof}
\begin{dfn}
For a cograph $G$ and $A\subset V(G)$, a vertex $x \in V(G)$ is a \emph{universal vertex} to $A \subset V(G)$, if $\forall v \in A, \{x,v\} \in E(G)$. An edge $\{x,y\} \in E(G)$ is a \emph{universal edge} to $A \subset V(G)$, if $ \forall v \in A, \{x,v\} \in E(G)$ or $\{y,v\} \in E(G)$.
\end{dfn}
\begin{lemma}
\label{lem2}
Let $G$ be a $k$-connected cograph and $S$ be a $k$-size minimal vertex separator in $G$.  Let $G_1, G_2, \ldots, G_k, k \geq 2$ be the connected components in $G \setminus S$. Then, every vertex $x \in S$ is universal to $V(G)\setminus S$.
\end{lemma}
\begin{proof}
It is enough to show that each vertex in $S$ is universal to each $G_i$. If $G_i$ is trivial, then the claim is true.  Suppose, $G_1$ is a non-trivial component in $G \setminus S$.  If, on the contrary, there exists a vertex $x$ in $S$ such that $x$ is not universal to $G_1$.  That is, there exists a vertex $y$ in $G_1$ such that $\{x,y\} \notin E(G)$.  Since $S$ is a minimal vertex separator there must exist $z \not = y$ in $G_1$ such that $\{x,z\} \in E(G)$.   Since $y$ and $z$ belongs to the same connected component $G_i$, there exists a path $P_{zy}=\{z=w_1,w_2,\ldots,w_k=y\}$ in $G_i$.   By \emph{Lemma \ref{lem1}}, $\{x,w_i\} \in E(G), 1 < i \leq k$.   This implies that $\{x,y\} \in E(G)$, which is a contradiction to our earlier observation.  Therefore, the claim follows.  \hfill\(\qed\)
\end{proof}
\begin{corollary}
\label{cor1}
Let $G$ be a $k$-connected cograph and $S$ be a $k$-size minimal vertex separator in $G$.  Then, every edge $\{u,v\}$ in $G \setminus S$ is universal to $S$.
\end{corollary}
\begin{proof}
By \emph{Lemma \ref{lem2}}, each vertex in $S$ is universal to each $G_i$.  It must be the case that every edge in $G_i$ is universal to $S$. \hfill\(\qed\)
\end{proof}
\begin{corollary}
\label{cor2}
Let $G$ be a $k$-connected cograph and $S$ be a $k$-size minimal vertex separator in $G$.  Then, each vertex $v$ in $V(G) \setminus S$ is universal to $S$.
\end{corollary}
\begin{proof}
Follows from \emph{Lemma \ref{lem1}} and \emph{Corollary \ref{cor1}}. \hfill\(\qed\)
\end{proof}
\subsection{Listing all minimal vertex separators in cographs}
We now present an algorithm to list all minimal vertex separators in cographs.  Our algorithm makes use of the underlying cotree and the structural properties presented in the previous section.
\begin{algorithm}[H]
\caption{Enumeration of all Minimal Vertex Separators in a Cograph}
\label{algmvs}
\begin{algorithmic}[1]
\STATE{\textbf{Input:} Cograph $G$, Cotree $T$}
\STATE{\textbf{Output:} All minimal vertex separators in $G$}
\STATE{$R$ be the root of $T$.  $N_T(R)=\{w_1,\ldots,w_t\}$, let $G_1,G_2,\ldots,G_t$ denote the subgraphs induced by the leaves in the subtrees rooted at $w_i$ in $T$.}
\STATE{\textbf{for} $i=1$ to $t$}
	\STATE{\hspace{0.5cm}$S_{i}:=V(G)\setminus V(G_i)$.}
	\STATE{\hspace{0.5cm}Output $S_{i}$ as a minimal vertex separator of $G$}
\end{algorithmic}
\end{algorithm}
\begin{lemma}
\label{mvslemma}
Given a cograph $G$, Algorithm \ref{algmvs} enumerates all minimal vertex separators in $G$.
\end{lemma}
\begin{proof}
Since $G$ is connected, the root node $R$ of $T$ is labelled $1$.   Observe that in any cotree $T$, the children of $R$ are labelled $0$.  Further, labels alternate between $1$ and $0$ as we move down from the root to leaf.  This implies that for all $i$, $G_{i}$ is disconnected. So, any $i$, $V(G) \setminus V(G_i)$ forms a vertex separator $S$.  Note that the set $V(G) \setminus V(G_i)$, on removal leaves the graph $G_i$ which is disconnected as the degree of $w_i$ is at least 2.  Observe that there does not exist a subset $S' \subset S$ such that $G \setminus S'$ is disconnected as every vertex in $G_i$ is adjacent to every vertex in $V(G) \setminus V(G_i)$.  So, the set $S$ output by our algorithm is minimal.   Since each $G_i$ yields a minimal vertex separator, our algorithm prints all minimal vertex separators.   Further, the algorithm runs in linear time.  \hfill\(\qed\)
\end{proof}
\subsection{Constrained vertex separators}
Given a connected graph $G$, a subset $S \subset V(G)$ is a \emph{connected vertex separator} if $S$ is a minimal vertex separator and $G[S]$, the graph induced on $S$, is connected.  If $G[S]$ is an independent set (stable set), then $S$ is a stable vertex separator.  It is known that finding a minimum connected vertex separator in general graphs, and in particular, in  chordality 5 graphs are  NP-complete \cite{narayanaswamy}.   In \cite{manogna}, it is shown that MIN-CONNECTED VERTEX SEPARATOR is polynomial-time solvable in $2K_2$-free graphs which are a strict subclass of chordality 5 graphs.  Finding a minimum stable vertex separator in general graphs is NP-complete \cite{dragan} and polynomial-time solvable restricted to triangle-free graphs and $2K_2$-free graphs \cite{manogna}.   In this paper, we shall present  polynomial-time algorithms for these problems in cographs which are also a strict subclass of chordality 5 graphs. \\

\noindent\textbf{Finding a minimum connected vertex separator:}\\
Note that any minimum connected vertex separator contains a minimal vertex separator as a subgraph.  Further, if the degree of $R$ in $T$ is at least 3, then each minimal vertex separator output by Algorithm \ref{algmvs} is indeed a minimum connected vertex separator in $G$.  Note that by the construction of $T$, any two $G_i$'s is connected, and since $S$ contains at least two $G_i$'s, $G[S]$ is connected.  If the degree of $R$ is two, then $S \cup \{x\}$, where $x \in V(G)\setminus S$ and $S$ is any minimum vertex separator, induces a minimum connected vertex separator in $G$.  This approach, also yields all minimum connected vertex separators in $G$, in linear time.   \\

\noindent\textbf{Finding a minimum stable vertex separator:}\\ Observe that if the degree of $R$ is at least 3,  then any minimal vertex separator $S$ in $G$ contains two $G_i$'s and hence $G[S]$ is not stable.  Therefore, if the degree of $R$ in $T$ is at least 3, then there is no stable vertex separator in $G$.   Let us consider the case where the degree of $R$ is two.  Let $T_1$ and $T_2$ denote the subtrees rooted at the two children of $R$ in $T$. A {\em $l$-star} is a tree on $l$ vertices with one vertex having degree $l-1$ and the other $l-1$ vertices have degree one.  We observe that a stable vertex separator in $G$ exists if and only if either $T_1$ or $T_2$ is a star.   Clearly, the complexity of this approach is linear in the input size.
\section{Vertex Connectivity Augmentation in Cographs}
\label{vcares}
We shall now present algorithms for vertex connectivity augmentation in cographs.  Further, we shall show that our algorithm is optimal by using lower bound arguments on the number of edges augmented.  We shall work with the following notation.  Let $G$ be a cograph and $T$ be its cotree.  Let $G_1,G_2,\ldots,G_t$ denote the subgraphs induced by the leaves in the subtrees rooted at $w_i$ in $T$, $w_i$ is a child of the root node of $T$.
\subsection{$(k+1)$-vertex connectivity augmentation}
\label{4.1}
Optimum version of $(k+1)$-vertex connectivity augmentation problem in cographs preserving the cograph property is formally defined as follows:
\begin{center}
\minibox[frame] {
\emph{\textbf{Instance:}} A $k$-vertex connected cograph $G$\\
\emph{\textbf{Question:}} Find a minimum cardinality augmentation set $E_{ca}$ such that $G \cup E_{ca}$ is a \\$(k+1)$-vertex connected cograph
}
\end{center}
\begin{lemma}
\label{lvca}
For every $G_i$ such that $|G_i|=n-k$, let $x_{i}$ be a vertex in $G_i$ such that $|\overline{N}_{G}(x_{i})|$ is minimum.  Then, any $(k+1)$-connectivity augmentation set $E_{ca}$ is such that $|E_{ca}| \geq \sum\limits_{x_{i}} |\overline{N}_{G}(x_{i})|$.
\end{lemma}
\begin{proof}
From  Lemma \ref{mvslemma}, we know that any minimal vertex separator $S$ is such that $S=V(G)\setminus V(G_i)$, for any $i$.  Since $G$ is a $k$-connected graph, $|S| \geq k$, and therefore, for any $i$, $|G_i| \leq n-k$.  Note that in any $(k+1)$-connected cograph, the size of any minimal vertex separator is at least $k+1$.  Therefore, to make $G$ a $(k+1)$-connected cograph $H$, in every $H_i$, we must have $|H_i| \leq n-k-1$ so that for all $i$, $|V(H) \setminus V(H_i)| \geq k+1$.  This implies that for each $G_i$ in $G$ such that $|G_i|=n-k$, we must remove a vertex $x$ from $G_i$ and include $x$ as a child of the root node.  Due to this modification, we must augment all edges from $x \in G_i$ to all vertices in $\overline{N}_{G}(x)$.  To ensure optimum, we remove $x_{i}$ from $G_i$ such that $|\overline{N}_{G}(x_{i})|$ is minimum.
Thus, any $(k+1)$-connectivity augmentation set $E_{ca}$ has atleast $\sum\limits_{x_{i}}|\overline{N}_{G}(x_{i})|$ edges.  \hfill\(\qed\)
\end{proof}
\begin{algorithm}[H]
\caption{$(k+1)$-vertex connectivity augmentation of a cograph}
\label{vca}
\begin{algorithmic}[1]
\STATE{\textbf{Input:} $k$-connected cograph $G$, cotree $T$}
\STATE{\textbf{Output:} $(k+1)$-connected cograph $H$ of $G$}
\STATE{\textbf{for} $i=1$ to $t$}
\STATE{\hspace{0.5cm}\textbf{if} $|G_i|=n-k$}
\STATE{\hspace{1cm} Find a vertex $x_i\in V(G_i)$ such that $|\overline{N}_{G}(x_i)|$ is minimum}
\STATE{\hspace{1cm} $\forall y \in \overline{N}_{G}(x_i)$, augment the edge $\{x_i,y\}$ to $G$ and update $E_{ca}$}
\STATE{Output the augmented graph $H$}
\end{algorithmic}
\end{algorithm}
{\bf Proof of correctness of Algorithm \ref{vca}:}
In \emph{Steps 4-5}, the algorithm finds all the subgraphs such that $|G_i|=n-k$ and finds a vertex $x_i\in G_i$ such that $|\overline{N}_{G}(x_i)|$ is minimum.  It further augments all the edges between $x_i$ and the vertices in $\overline{N}_{G}(x_i)$ to the augmentation set as given in \emph{Step 6}. Therefore, the algorithm augments $\sum\limits_{x_{i}}|\overline{N}_{G}(x_{i})|$ edges in total.  For every $G_i$ in $G$ such that $|G_i|=n-k$, let $S=V(G)\setminus V(G_i)$ be the minimum vertex separator. Because we remove a vertex $x_i$ from every such $G_i$, in $G\cup E_{ca}$, $S\cup \{x_i\}$ becomes a minimum vertex separator. Therefore, the resultant graph is a $(k+1)$-connected cograph.  Further, the algorithm runs in $O(n)$ time.

\subsection{Weighted $(k+1)$-vertex connectivity augmentation}
\label{4.2}
Optimum version of weighted $(k+1)$-vertex connectivity augmentation problem in cographs preserving the cograph property is formally defined as follows:
\begin{center}
\minibox[frame]{
\emph{\textbf{Instance:}} A $k$-vertex connected cograph $G$ and a weight function $w:E(\overline{G})\rightarrow R^{+}$\\
\emph{\textbf{Question:}} Find a set $E_{wca}$ such that $W_{wca}=\sum\limits _{(u,v)\in E_{wca}}w(u,v)$ is minimum and $G\cup E_{wca}$ is a\\ $(k+1)$-vertex connected cograph
}
\end{center}
\begin{lemma}
For every $G_i$ such that $|V(G_i)|=n-k$, let $x_{i}\in G_i$ be a vertex such that $W(x_{i})=\sum\limits_{y\in \overline{N}_{G}(x_{i})}w(x_{i},y)$ is minimum. Then, any weighted $(k+1)$-connectivity augmentation set $E_{wca}$ is such that $|W_{wca}|\geq \sum\limits_{x_{i}}W(x_{i})$.
\end{lemma}
\begin{proof}
Similar to the proof of Lemma \ref{lvca}, to make $G$ a $(k+1)$-connected graph, for every $G_i$ such that $|G_i|=n-k$, we remove a vertex $x$ from $G_i$ and include them as a child of the root node in $T$.  While doing so,  we augment edges from $x$ to all the vertices in $\overline{N}_{G}(x)$ so that $G_i$ has $n-k-1$ vertices.   To ensure optimality, $x_{i}$ is a vertex in $G_i$ such that $W(x_{i})=\sum\limits_{y\in \overline{N}_{G}(x_{i})}w(x_{i},y)$ is minimum and therefore, the weight of any $E_{wca}$ is atleast $\sum\limits_{x_{i}}W(x_{i})$.  In $G \cup E_{wca}$, every $x_i$ is universal to $V(G)\setminus \{x_i\}$ which implies that all $x_i$ become the children of $R$ in $T[G\cup E_{wca}]$. Thus, the cotree property is preserved. This completes the proof of the lemma.
\end{proof}
\begin{algorithm}[ht]
\caption{Weighted $(k+1)$-vertex Connectivity Augmentation of a Cograph}
\label{wvca}
\begin{algorithmic}[1]
\STATE{\textbf{Input:} $k$-connected cograph $G$, cotree $T$, weight function $w:E(\overline{G})\rightarrow R^{+}$}
\STATE{\textbf{Output:} $(k+1)$-connected cograph $H$ of $G$}
\STATE{\textbf{for} $i=1$ to $t$}
\STATE{\hspace{0.5cm}\textbf{if} $|G_i|=n-k$}
\STATE{\hspace{1cm}Find a vertex $x_i\in V(G_i)$ such that $W(x_i)=\sum\limits_{y\in \overline{N}_{G}(x_i)}w(x_i,y)$ is minimum}
\STATE{\hspace{1cm}$\forall y\in \overline{N}_{G}(x_i)$, augment the edge $\{x_i,y\}$ to $G$ and update $E_{wca}$}
\STATE{Output the augmented graph $H$}
\end{algorithmic}
\end{algorithm}
{\bf Proof of correctness of Algorithm \ref{wvca}:}
In \emph{Steps 4-5}, the algorithm finds all the subgraphs such that $|G_i|=n-k$ and finds a vertex $x_i\in G_i$ such that $W(x_i)=\sum\limits_{y\in \overline{N}_{G}(x_i)}w(x_i,y)$ is minimum.  It further augments all the edges between $x_i$ and the vertices in $\overline{N}_{G}(x_i)$ to the augmentation set in \emph{Step 6}. Therefore, the algorithm augments edges with weight $\sum\limits_{x_{i}}W(x_{i})$. Let $S=V(G)\setminus V(G_i)$ be the minimum vertex separator, for every $G_i$ in $G$ such that $|G_i|=n-k$.  Because we remove a vertex $x_i$ from every such $G_i$, in $G\cup E_{wca}$, $S\cup \{x_i\}$ becomes the minimum vertex separator.  In $T[G\cup E_{wca}]$, every $x_i$ becomes a child of $R$. Thus, the resultant graph is a $(k+1)$-connected cograph and our algorithm is linear in the input size.\\ \\
{\tt Remark:} Results presented in Section \ref{4.2} are a generalization of results presented in Section \ref{4.1}.

\section{Edge Connectivity Augmentation in Cographs}
\label{ecares}
In this section, we shall discuss two variants of edge connectivity augmentation problems in cographs.  For a connected graph $G$, a set $F\subset E(G)$ is called an \emph{edge separator} if $G-F$ is disconnected and $F$ is a \emph{minimum edge separator} if it is an edge separator of least size. The \emph{edge connectivity} of $G$ refers to the size of a minimum edge separator. A connected graph $G$ is said to be \emph{$k$-edge connected} if its edge connectivity is $k$.
\begin{lemma}
\label{minedgsep}
Let $G$ be a cograph. Then, any minimum edge separator $F$ in $G$ is such that $|F|=\delta (G)$, where $\delta (G)$ refers to the minimum degree of $G$.
\end{lemma}
\begin{proof}
Any edge separator in a graph $G$ is obtained by removing all the edges between some $A\subset V(G)$ and $V(G)\setminus A$.
\begin{description}
\item[\textbf{Case 1: }] $|A|=1$. Clearly, by removing edges incident on the minimum degree vertex, the graph is disconnected.  Thus, $|F|=\delta (G)$.
\item[\textbf{Case 2: }] $|A|\geq 2$. Let $y$ denote the cardinality of edge separator when $|A|\geq 2$ and $z$ denote the degree of some vertex $v\in A$, respectively. To prove the claim, we show that $y\geq z$. Consequently, it follows that minimum edge separator in $G$ can be obtained when $|A|=1$. Let $G_1,G_2,\ldots G_t$ denote the induced subgraphs of $G$ on the leaves of the subtrees rooted at the children of the root node in $T$. For all $1\leq i\leq t$, let $X_{i}=V(G_{i})\cap A$. Let $X_1\leq |X_2|\leq \ldots \leq |X_{t}|$ and $v\in X_1$. Clearly, $y=|X_1|(n-|G_1|-|A|+|X_1|)+\sum\limits_{i=1}^{|X_1|}d_{1i}+\ldots +|X_t|(n-|G_t|-|A|+|X_t|)+\sum\limits_{i=1}^{|X_t|}d_{ti}$, where $d_{ti}=|N_{G_i\setminus X_i}(x_{ti})|$ and $x_{ti}$ is the $i$th element in $X_t$. Suppose $v=x_{11}$. Degree of $v$ can at most be $n-|G_1|+|X_1|-1+d_{11}$. On the contrary, assume that $|X_1|(n-|G_1|-|A|+|X_1|)+\sum\limits_{i=1}^{|X_1|}d_{1i}+\ldots +|X_t|(n-|G_t|-|A|+|X_t|)+\sum\limits_{i=1}^{|X_t|}d_{ti}<n-|G_1|+|X_1|-1+d_{11}$ which implies $(|X_1|-1)(n-|G_1|-|A|+|X_1|-1)-|A|+|X_1|+\sum\limits_{i=2}^{|X_1|}d_{1i}+|X_2|(n-|G_2|-|A|+|X_2|)+\sum\limits_{i=1}^{|X_2|}d_{2i}+\ldots +|X_t|(n-|G_t|-|A|+|X_t|)+\sum\limits_{i=1}^{|X_t|}d_{ti}<0$ which is a contradiction. Therefore, cardinality of the edge separator when $|A|\geq 2$ is greater than or equal to the cardinality of the edge separator when $|A|=1$.
\end{description}
Hence, size of any minimum edge separator in a cograph is $\delta (G)$. This completes the proof of the lemma.\hfill\(\qed\)
\end{proof}

\subsection{$(k+1)$-edge connectivity augmentation}

Optimum version of $(k+1)$-edge connectivity augmentation problem in cographs is formally defined as follows:
\begin{center}
\minibox[frame]{
\emph{\textbf{Instance:}} A $k$-edge connected cograph $G$\\
\emph{\textbf{Solution:}} A minimum cardinality augmentation set $E_{ca}$ such that $G\cup E_{ca}$ is a \\$(k+1)$-edge connected graph
}
\end{center}

For a connected graph $G$, a set of edges $E_{c}\subseteq E(G)$ forms an \emph{edge cover} if every vertex of $G$ is incident with at least one edge in $E_{c}$. An edge cover with minimum cardinality is known as \emph{minimum edge cover}. Let $\rho (G)$ denote the cardinality of the minimum edge cover in the graph $G$. For a disconnected graph $G$, let $G_1,G_2,\ldots ,G_k,k\geq 2$ be the connected components. For a trivial component $G_i$, let $\rho (G_i)=1$. Then, we define $\rho (G)=\sum\limits_{i}\rho (G_i)$.

\begin{lemma}
\label{leca}
Let $G$ be a $k$-connected cograph and let $X$ denote the set of $k$-degree vertices in $G$. Then, any $(k+1)$-edge connectivity augmentation set $E_{ca}$ is such that $E_{ca}\geq \rho (\overline{G}[X])$.
\end{lemma}
\begin{proof}
From Lemma \ref{minedgsep}, cardinality of any minimum edge separator is equal to $\delta (G)$. Therefore, to make $G$ a $(k+1)$-connected graph, we must have $\delta (G)=k+1$. This implies, we must increase degree of every $k$-degree vertex atleast by one. Removing any edge from an edge cover leaves an uncovered vertex which implies every edge in the edge cover has one vertex with degree one. Hence, every edge cover has minimum degree one. If $\overline{G}[X]$ is connected, minimum number of edges required to increase degree of every $k$-degree vertex atleast by one is equal to $\rho (\overline{G}[X])$. And if $\overline{G}[X]$ is disconnected, and the connected component is trivial, then we must add an edge from that vertex to some non-adjacent vertex in $G$. If the component is non-trivial, then we must augment atleast cardinality of minimum edge cover number of edges in that component. Therefore, we must augment $\rho (\overline{G}[X])$ number of edges in total. Hence, any $(k+1)$-edge connectivity augmentation set has atleast $\rho (\overline{G}[X])$ edges. This completes the proof of the lemma.\hfill\(\qed\)
\end{proof}

\subsubsection{Outline of the Algorithm}
Our algorithm first finds the set $X$ containing all the $k$-degree vertices in $G$. If $\overline{G}[X]$ is connected, it finds minimum edge cover in $\overline{G}[X]$ and adds all the edges in the edge cover to the augmentation set. If it is disconnected, it traverses through each connected component. If the connected component is trivial, it augments an edge between that vertex and some non-adjacent vertex in $G$. If the connected component is non-trivial, it finds minimum edge cover and augments all the edges in the edge cover.

\subsubsection{The Algorithm}\hfill \\ \\
We now present an algorithm for $(k+1)$-edge connectivity augmentation and further prove that our algorithm is optimal.

\begin{algorithm}[H]
\caption{$(k+1)$-edge Connectivity Augmentation of a Cograph}
\label{eca}
\begin{algorithmic}[1]
\STATE{\textbf{Input:} $k$-connected cograph $G$}
\STATE{\textbf{Output:} $(k+1)$-connected graph $H$ of $G$}
\STATE{Let $X$ be the set of $k$-degree vertices in $G$}
\STATE{\textbf{if} $\overline{G}[X]$ is connected \textbf{then}}
\STATE{\hspace{0.5cm}Find minimum edge cover $E_c$ in $\overline{G}[X]$}
\STATE{\hspace{0.5cm}$\forall \{u,v\}\in E_c$, augment the edge $\{u,v\}$ to $G$ and update $E_{ca}$}

\STATE{\textbf{else}}
\STATE{\hspace{0.5cm}\textbf{for} each connected component $G_{i}$ in $\overline{G}[X]$}
\STATE{\hspace{1cm}\textbf{if} $|G_{i}|=1$ and $x\in V(G_i)$ \textbf{then}}
\STATE{\hspace{1.5cm}Find a vertex $y\in \overline{N}_{G}(x)$ and augment the edge $\{x,y\}$ to $G$ and update $E_{ca}$}
\STATE{\hspace{1cm}\textbf{else}}
\STATE{\hspace{1.5cm}Find minimum edge cover $E_c$ in $G_{i}$}
\STATE{\hspace{1.5cm}$\forall \{u,v\}\in E_c$, augment the edge $\{u,v\}$ to $G$ and update $E_{ca}$}
\STATE{Output the augmented graph $H$}
\end{algorithmic}
\end{algorithm}

\subsubsection{Proof of correctness of Algorithm \ref{eca}}
In \emph{Step 1}, the algorithm finds the set $X$ containing the set of $k$-degree vertices in $G$. If $\overline{G}[X]$ is connected, it finds minimum edge cover and adds all the edges to the augmentation set in \emph{Steps 4-6}. This ensures degree of all vertices in $X$ is increased atleast by one. If $\overline{G}[X]$ is disconnected, it traverses through all connected components. If the component is trivial, it augments one edge from that vertex to a non-adjacent vertex in $G$ in \emph{Steps 9-10}. If the component is trivial, it finds minimum edge cover in that component and adds those edges in \emph{Steps 12-13} which implies degree of all those vertices is also increased atleast by one. Since finding minimum edge cover can be done in $O(n)$ time in cographs \cite{yu}, where $n$ is the size of the parse tree, our algorithm also take $O(n)$ time. \\

\subsection{Weighted $(k+1)$-edge Connectivity Augmentation}

Optimal version of weighted $(k+1)$-edge connectivity augmentation problem in cographs is formally defined as follows:
\begin{center}
\minibox[frame]{
\emph{\textbf{Instance:}} A $k$-edge connected cograph $G$ and a weight function $w:E(\overline{G})\rightarrow R^{+}$\\
\emph{\textbf{Solution:}} An augmentation set $E_{wca}$ such that $W_{wca}=\sum\limits _{(u,v)\in E_{wca}}w(u,v)$ is minimum and $G\cup E_{wca}$ is a\\ $(k+1)$-edge connected graph
}
\end{center}

For a connected weighted graph $G$, a minimum weighted set of edges $E_{c}^{w}\subseteq E(G)$ forms an \emph{minimum weighted edge cover} if every vertex of $G$ is incident with at least one edge in $E_{c}^{w}$. For the graph $G$, let $\rho _w(G)$ denote the weight of the minimum weighted edge cover. For a disconnected graph $G$, let $G_1,G_2,\ldots ,G_k,k\geq 2$ be the connected components. For a trivial component $G_i$ and $V(G_i)=\{x\}$, let $\rho _w(G_i)=w(x,y)$, where $w(x,y)=min\{w(x,y)\forall y\in \overline{N}_{G}(x)\}$. Then, we define $\rho _w(G)=\sum\limits_{i}\rho _w(G_i)$.

\begin{lemma}
Let $G$ be a $k$-connected cograph and let $X$ denote the set of $k$-degree vertices in $G$. Then, any weighted $(k+1)$-edge connectivity augmentation set $E_{wca}$ is such that $W_{wca}\geq \rho _w(\overline{G}[X])$.
\end{lemma}
\begin{proof}
Similar to the proof of Lemma \ref{leca}, to make $G$ a $(k+1)$-connected graph, we must increase degree of every $k$-degree vertex atleast by one. If $\overline{G}[X]$ is connected, then to increase degree of every $k$-degree vertex atleast by one we must augment edges with atleast $\rho _w(\overline{G}[X])$ weight. And if $\overline{G}[X]$ is disconnected, and the connected component is trivial, then we must add an edge with least weight from that vertex to some non-adjacent vertex in $G$. If the component is non-trivial, then we must augment edges with atleast weight of minimum weighted edge cover in that component. Therefore, we must augment edges with total weight of $\rho _w(\overline{G}[X])$. Hence, any weighted $(k+1)$-edge connectivity augmentation set has atleast $\rho _w(\overline{G}[X])$ weight. This completes the proof of the lemma. \hfill\(\qed\)
\end{proof}

\subsubsection{Outline of the Algorithm}
Our algorithm first finds the set $X$ containing all the $k$-degree vertices in $G$. If $\overline{G}[X]$ is connected, it finds minimum weighted edge cover in $\overline{G}[X]$ and adds all the edges in the edge cover to the augmentation set. If it is disconnected, it traverses through each connected component. If the connected component is trivial, it augments an edge with minimum weight between that vertex and some non-adjacent vertex in $G$. If the connected component is non-trivial, it finds minimum weighted edge cover and augments all the edges in the edge cover.

\subsubsection{The Algorithm}\hfill \\ \\
We now present an algorithm for weighted $(k+1)$-edge connectivity augmentation and further give proof of correctness of the algorithm.

\begin{algorithm}[H]
\caption{Weighted $(k+1)$-edge Connectivity Augmentation of a Cograph}
\label{weca}
\begin{algorithmic}[1]
\STATE{\textbf{Input:} $k$-connected cograph $G$, weight function $w:E(\overline{G})\rightarrow R^{+}$}
\STATE{\textbf{Output:} $(k+1)$-connected graph $H$ of $G$}
\STATE{Let $X$ be the set of $k$-degree vertices in $G$}
\STATE{\textbf{if} $\overline{G}[X]$ is connected \textbf{then}}
\STATE{\hspace{0.5cm}Find minimum weighted edge cover $E_c^w$ in $\overline{G}[X]$}
\STATE{\hspace{0.5cm}$\forall \{u,v\}\in E_c^w$, augment the edge $\{u,v\}$ to $G$ and update $E_{wca}$}

\STATE{\textbf{else}}
\STATE{\hspace{0.5cm}\textbf{for} each connected component $G_{i}$ in $\overline{G}[X]$}
\STATE{\hspace{1cm}\textbf{if} $|G_{i}|=1$ and $x\in V(G_i)$ \textbf{then}}
\STATE{\hspace{1.5cm}Find a vertex $y\in \overline{N}_{G}(x)$ and augment the edge $\{x,y\}$ to $G$ and update $E_{wca}$}
\STATE{\hspace{1cm}\textbf{else}}
\STATE{\hspace{1.5cm}Find minimum weighted edge cover $E_c^w$ in $G_{i}$}
\STATE{\hspace{1.5cm}$\forall \{u,v\}\in E_c^w$, augment the edge $\{u,v\}$ to $G$ and update $E_{wca}$}
\STATE{Output the augmented graph $H$}
\end{algorithmic}
\end{algorithm}

\subsubsection{Proof of correctness of Algorithm \ref{weca}}
The algorithm finds the set $X$ containing the set of $k$-degree vertices in $G$ in \emph{Step 1}. If $\overline{G}[X]$ is connected, it finds minimum weighted edge cover and adds all the edges to the augmentation set in \emph{Steps 4-6}. This ensures degree of all vertices in $X$ is increased atleast by one. If $\overline{G}[X]$ is disconnected, it traverses through all connected components. If the component is trivial, it augments the edge with least weight from that vertex to a vertex in $G$ in \emph{Steps 9-10}. If the component is non-trivial, it finds minimum weighted edge cover in that component and adds those edges in \emph{Steps 12-13} which implies degree of all vertices in that component is also increased atleast by one. Since minimum weighted edge cover problem is open in cographs, we use the fastest general graph minimum weighted edge cover algorithm \cite{gabow} that runs in $O(mn+n^2logn)$ time. Thus, time complexity of Algorithm \ref{weca} is $O(mn+n^2logn)$.

\section{Some NP-hard Problems in Cographs}
\label{dpres}
\noindent In this section, we present a generic framework using dynamic programming paradigm to solve three optimization problems; the longest path, Steiner path and minimum leaf spanning tree problems.  We work with the {\em parse tree} of a cograph.  The parse tree is similar to a cotree, which is a binary tree and helps in the design of dynamic programming based algorithms.

In our approach, we traverse the parse tree in post order traversal and maintain some states at each node in the parse tree which we update recursively.  Our main idea is to find an optimal solution at every node in the parse tree by combining the optimal solutions of its children as we traverse the parse tree.  Throughout this section, let $T$ denote the parse tree constructed from the input cograph $G$. In each case study, the update at a node $v$ is done recursively depending upon whether $v$ is a leaf node, $v$ is labelled 0 or $v$ is labelled 1.   We shall present the process in the respective sections to update the states in all the three cases.  When the algorithm terminates, we have the final solution to the problem stored at the root node of the parse tree.
\subsection{The Longest Path Problem}
Hamiltonian path (cycle) is a well-known problem in graph theory with many practical applications in the field of computing.  Given a connected graph $G$, the Hamiltonian path (cycle) problem asks for a spanning path (cycle) in $G$.   This problem is NP-complete in general and in special graphs such as chordal graphs and chordal bipartite graphs.   Polynomial-time algorithms for this problem are known in interval graphs \cite{ioannidou}, cocomparability graphs \cite{mertzios} and bipartite permutation graphs \cite{uehara}.  In this paper, we present a polynomial-time algorithm for the longest path problem which is a generalization of the Hamiltonian path problem.

For a cograph $G$, we work with the parse tree $T$.  For a node $v$ in $T$, $T_v$ denotes the subtree rooted at $v$ and $G_v$ denotes the underlying cograph corresponding to $T_v$.  We maintain two states $P_v$ and $U_v$ for every node $v \in T$, where $P_v$ is the longest path in the graph $G_v$.  While updating $P_v$, we make use of paths generated by recursive subproblems. The paths that are not used for updating $P_v$ are included in $U_v$ which may be used later for updating ancestors of $v$ in $T$. Let $P_1=(x_1,\ldots ,x_{|P_1|})$, $U_1=(Q_1,Q_2,\ldots, Q_{|U_2|})$ and $P_2=(y_1,\ldots ,y_{|P_2|})$, $U_2=(R_1,R_2,\ldots, R_{|U_2|})$ denote the states w.r.t. the first and second child of $v$, respectively in $T$. Let $r_{jk}$ denote the $k$th vertex in the path $R_j$ in $U_2$. The states $P_v$ and $U_v$ are updated as follows.
\begin{enumerate}
\item When $v$ is a leaf node, $P_v$ contains the vertex $v$ and $U_v$ is empty.

\item When $v$ is labelled 0, without loss of generality, let $|P_1|\geq |P_2|$. Now, $P_v$ contains the path $P_1$ and $U_v$ contains the set of paths in  $U_1$, $U_2$ and $P_2$. Let $U_v=(S_1,S_2,\ldots ,S_{|U_v|})$ be the ordering of the paths in $U_v$ such that $|S_1|\geq |S_2|\geq \ldots \geq |S_{|U_v|}|$.

\item When $v$ is labelled 1, without loss of generality, assume $|U_1|\geq |U_2|$. Initialize $P_v$ to $\emptyset$ and assume all paths in $U_1$ and $U_2$ are {\em uncovered} initially.  A path in $U_i$ is said to be {\em covered} if it is considered as part of update.  Let $a$ and $b$ denote the number of paths uncovered in $U_1$ and $U_2$, respectively w.r.t. $P_v$.  Let $c$ denote the number of vertices uncovered in $P_2$ w.r.t. $P_v$. Firstly, we extend the path $P_v$ by concatenating a path in $U_1$ and a vertex of a path in $U_2$, that is, the end point of a path in $U_1$ is attached to a vertex of a path in $U_2$.  We do this alternately (a path in $U_1$ and a vertex in a path of $U_2$ by preserving the order of vertices in the path in $U_2$) until $a=b+1$ or $b=0$.  Once we exhaust vertices in a path in $U_2$, the next path in $U_2$ is considered.  If $a=b+1$, then we extend $P_v$ by concatenating a path in $U_1$ and a path in $U_2$ preserving the order until $a=0$. Then, we extend $P_v$ by adding $P_2$, and further extend by adding $P_1$.  Otherwise, we extend the path $P_v$ by concatenating a path in $U_1$ and a vertex in $P_2$ preserving the order until $a=0$ or $c=0$.  Similar to the above, while extending we alternate between a path in $U_1$ and a vertex of a path in $U_2$.  If $c=0$, we further extend the path $P_v$ by including $P_1$ and $P_2$.
\item At the end, the update is done for the root node, $P_v$ which stores the longest path in the input cograph $G$.
\end{enumerate}
\subsubsection{The Algorithm}\hfill \\ \\
We shall now present an algorithm for finding a longest path in a cograph and further prove that our algorithm is optimal.
\begin{algorithm}[ht]
\caption{Longest path in a cograph}
\label{lp}
\begin{algorithmic}[1]
\STATE{\textbf{Input:} A cograph $G$, parse tree $T$}
\STATE{\textbf{Output:} The longest path $P$ in $G$}
\STATE{Update is done by visiting nodes in $T$ in post order traversal}
\STATE{\textbf{if} $v$ is a leaf node \textbf{then}}
	\STATE{\hspace{0.5cm}$P_v=(v)$ and $U_v=\emptyset$}
\STATE{\textbf{else}}
	\STATE{\hspace{0.5cm}\texttt{/* Let $P_1=(x_1,\ldots ,x_{|P_1|})$, $U_1=(Q_1,Q_2,\ldots, Q_{|U_2|})$ and $P_2=(y_1,\ldots ,y_{|P_2|})$, $U_2=(R_1,R_2,\ldots, R_{|U_2|})$ denote the states w.r.t. the first and second child of $v$, respectively in $T$. Let $r_{jk}$ denote the $k$th vertex in the path $R_j$ in $U_2$. */}}
	\STATE{\hspace{0.5cm}\textbf{if} $v$ is labelled 0 \textbf{then}}
		\STATE{\hspace{1cm}$P_v\leftarrow P_1$ \texttt{/* Assume $P_1=max(P_1,P_2)$ */}}
		\STATE{\hspace{1cm}$U_v\leftarrow U_1\cup U_2\cup P_2$. Let $U_v=(S_1,S_2,\ldots ,S_{|U_v|})$ be the ordering of the paths in $U_v$ such that $|S_1|\geq |S_2|\geq \ldots \geq S_{|U_v|}$}
	\STATE{\hspace{0.5 cm}\textbf{else}}
		\STATE{\hspace{1cm}Initialize $P_v=\emptyset$, $a=|U_1|$, $b=|U_2|$, $c=|P_2|$ and $i=j=k=1$ \texttt{/* Assume $|U_1|\geq |U_2|$ */}}
		\STATE{\hspace{1cm}\textbf{while} $a>b+1$ and $b>0$}
		\STATE{\hspace{1.5 cm}$P_v=(P_v,Q_i,r_{jk})$; $i=i+1$; $k=k+1$; $a=a-1$}
		\STATE{\hspace{1.5cm}\textbf{if} $k=|R_j|+1$ \textbf{then} $j=j+1$; $k=1$; $b=b-1$}
		\STATE{\hspace{1cm}\textbf{if} $a=b+1$ \textbf{then} $P_v=(P_v,Q_i,R_j,Q_{i+1},R_{j+1},\ldots ,Q_{|U_2|},R_{|U_2|},Q_{|U_1|},P_2,P_1)$ and $U_v=\emptyset$}
			\STATE{\hspace{1cm}\textbf{else}}
			\STATE{\hspace{1.5cm}\textbf{while} $i<=|U_1|$ and $c>0$}
		\STATE{\hspace{2cm}$P_v=(P_v,Q_i,y_c)$; $i=i+1$; $c=c-1$}
		\STATE{\hspace{1.5cm}$P_v=(P_v,P_1,P_2)$ and $U_v=(Q_i,Q_{i+1},\ldots ,Q_{|U_1|})$}
\STATE{Output $P_v$, the longest path in $G$}

\end{algorithmic}
\end{algorithm}

\subsubsection{Proof of correctness of Algorithm \ref{lp}}
To show that our algorithm indeed outputs the longest path in $G$, it is enough if we show that for every node $v\in T$, $P_v$ gives the longest path in the graph $G_v$ and $U_v$ contains the paths generated by recursive subproblems that are not used for updating $P_v$. We shall prove the claim by induction on the height of $T_v$, $T$ rooted at vertex $v$. \\
\emph{Basis Step:} When $v$ is a leaf node, the claim is true.\\
\emph{Induction Hypothesis:} Assume that the claim is true for a parse tree $T$ of height $h \leq k$.  Let $v$ be a node at height $k+1$, $k \geq 0$.  Let $G_1$ and $G_2$ denote the subgraphs induced by the leaves in the subtrees rooted at the children of the node $v$ in $T$. Let $P_1$ and $P_2$ be the longest paths given by our algorithm in $G_1$ and $G_2$, respectively. Let $U_1$ and $U_2$ contains the paths generated by recursive subproblems that are not used for updating $P_1$ and $P_2$, respectively. We shall now prove that the claim is true for the node $v$.\\
\emph{Induction Step:} Let $P_1=(x_1,\ldots ,x_{|P_1|})$, $P_2=(y_1,\ldots ,y_{|P_2|})$. Let $U_1=\{Q_1,Q_2,\ldots, Q_{|U_2|}\}$ and $U_2=\{R_1,R_2,\ldots, R_{|U_2|}\}$. Let $r_{jk}$ denote the $k$th vertex in the path $R_j$. When $v$ is labelled 0, $G_1$ and $G_2$ are not connected, so $max(P_1,P_2)$ (say $P_1$) is the longest path in $G_v$. Thus, the longest path given by our algorithm is correct in this case.  Further, $U_v$ contains the paths in $U_1$, $U_2$ and $P_2$. All the paths in $U_v$ are sorted in the decreasing order of their lengths. Note that none of the vertices in $U_1$ and $U_2$ are connected to $r_{j1}$ and $r_{j|R_j|}$ because of the node $v$ labelled 0. Also, none of the vertices in $U_1$ and $U_2$ are not connected to $x_1$, $x_{|P_1|}$ and $y_1$, $y_{|P_2|}$, respectively. \\
When $v$ is labelled 1, we do step analysis to prove the claim. Recall that every vertex in $G_1$ is connected to every vertex in $G_2$. Let $a$ and $b$ denote the number of uncovered paths in $P_v$ w.r.t. $U_1$ and $U_2$, respectively. In \emph{Steps 13-15} of the algorithm, we extend $P_v$ by concatenating the path $Q_i$ and the vertex $r_{jk}$ alternately until $a=b+1$ or $b=0$. If $a=b+1$, then we extend $P_v$ by concatenating the path $Q_i$ and the path $R_j$ until $a=0$. Further, we extend the path by adding $P_2$ followed by $P_1$ as in \emph{Step 16}. Therefore, $P_v$ contains the all of $V(G_v)$.  Hence, the longest path given by our algorithm is correct and $U_v$ contains the paths uncovered in $P_v$ w.r.t. $U_1$. Let $c$ be the number of vertices uncovered in the path $P_v$ w.r.t. $P_2$. Suppose $b=0$, then extend the path $P_v$ by concatenating the path $Q_i$ and the vertices in $P_2$ until $c=0$ or $a=0$ as in \emph{Steps 18-19}.  We further extend $P_v$ by
adding $P_1$ as given in \emph{Step 20}.  Now, if $a=0$, then our algorithm outputs a spanning path which implies that our algorithm is correct. Else, consider any path $P'$ of $G_{v}$. Since all paths in $U_v$ are sorted in decreasing order of lengths, to prove the claim, we show that the number of paths not part of $G_1$ with respect to $P'$ in $G_{v}$ is larger than the longest path $P_v$ enumerated by our algorithm.  Let $l$ be the number of times paths in $G_1$ and $G_2$ alternates in $P'$.  By the induction hypothesis, any path in $G_1$ must have at least $|U_1|$ paths not part of $P_v$.   Therefore, the number of paths not part of $G_1$ with respect to $P'$ $\geq |U_1|-l$ $\geq |U_1|-|G_2|$ $\geq |U_1|-|G_2|-1=a$. Thus, $|P_v|\geq |P'|$. Hence, $P_v$ is the longest path in $G_{v}$ and $U_v=U_1$.   This completes the induction argument.

\subsubsection{Trace of Algorithm \ref{lp}}\hfill \\ 
We now trace the steps of Algorithm \ref{lp} in the figure \ref{fig:lp}. We traverse through the nodes in the post order traversal in the parse tree. In this example, we traverse in the order a,b,c,1,d,2,3,e,f,4,g,h,i,5,6,7,j,k,8,l,m,9,10,n,o,11,12,13,14. While traversing, we update $P_v$ and $U_v$ as shown in the figure. For example, the states at nodes a,b,c,d,e,f,g,i,j,k,l,m,n,o are updated as per \emph{Step 5} in the algorithm as they are leaf nodes. And the states at nodes 2,3,5,7,8,9,12 are updated as per \emph{Steps 9-10}. The states at nodes 1,4,6,10,11,13,14 are updated following the \emph{Steps 12-20}.  

\begin{figure}[h!]
    \begin{center}
    \includegraphics[height=8cm,width=16cm]{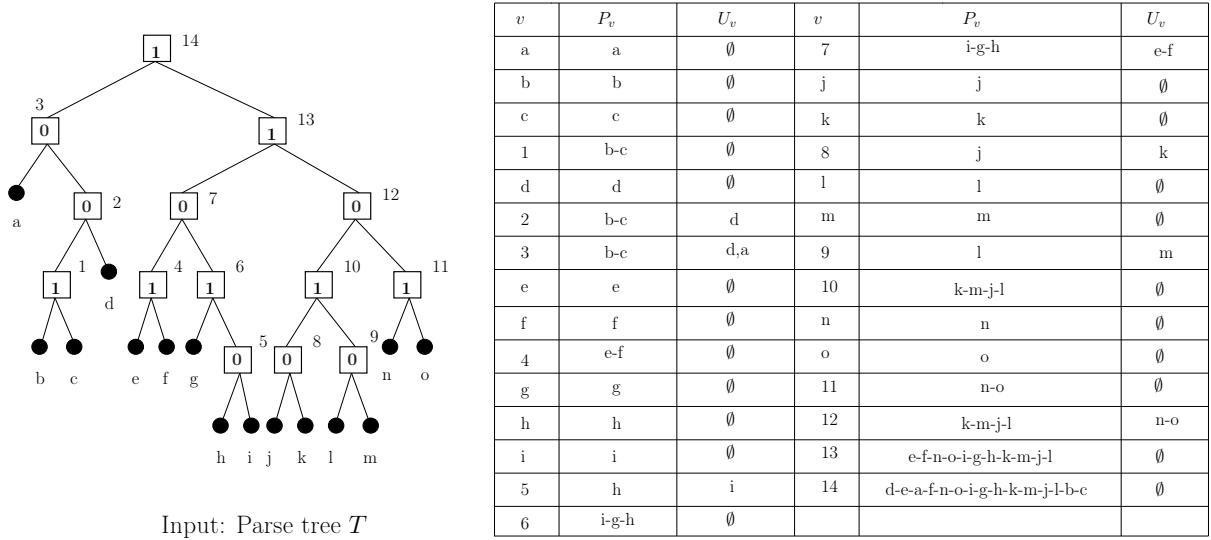}
    \caption{Trace of longest path algorithm when Hamiltonian path exists}
    \label{fig:lp}
    \end{center}    
\end{figure}

\begin{theorem}
The longest path problem in cographs is polynomial-time solvable.
\end{theorem}
\begin{proof}
Follows from the discussion presented in the previous section.\hfill\(\qed\)
\end{proof}

A simple path that visits all the vertices in $G$ is called the Hamiltonian path. A Hamiltonian cycle is a cycle that visits all the vertices in $G$ exactly once.  For every $v \in T$, let $G_1$ and $G_2$ denote the subgraphs induced by the children of the node $v$ in $T$. Let $P_1$ and $P_2$ be the longest paths given by our algorithm in $G_1$ and $G_2$, respectively. Let $U_1$ and $U_2$ contains the paths generated by recursive subproblems that are not used for updating $P_1$ and $P_2$, respectively.  We now give a necessary and sufficient condition for the existence of Hamiltonian path (cycle) in cographs.  Our result is based on the states defined as part of the longest path algorithm.   Note that the Hamiltonian path is a special case of the longest path problem.   We shall work with the notation used in this section to present our results.
\begin{theorem}
Let $G$ be a cograph and $T$ be its corresponding parse tree rooted at $v$. $G$ has a Hamiltonian path if and only if $|U_1| \leq |G_2|$.
\end{theorem}
\begin{proof}
\emph{Necessity: } If, on the contrary, assume that $|U_1|=|G_2|+1$.  In our algorithm, we concat paths in $U_1$ with vertices in $G_2$ alternately while constructing the longest path.  We further extend the path by including $P_1$.  Then, we have one path uncovered in $P_v$ w.r.t. $U_1$ which is a contradiction to the definition of Hamiltonian path.\\
\emph{Sufficiency: }$|U_1| \leq |G_2|$. We concatenate paths in $U_1$ with the vertices in $G_2$ alternately while constructing the longest path. Thus, we cover all the paths in $U_1$. Finally, we concatenate $P_1$ to the path. Therefore, it follows that $G_v$ has a Hamiltonian path.\qed
\end{proof}
\begin{theorem}
Let $G$ be a cograph and $T$ be its corresponding parse tree rooted at $v$.  $G$ has a Hamiltonian cycle if and only if $|U_1| \leq |G_2|-1$.
\end{theorem}
\begin{proof}
\emph{Necessity:} The proof is similar to the Hamiltonian path problem. \qed

\end{proof}
\subsection{The Steiner Path Problem}
While the longest path problem is one form of generalization of the Hamiltonian problem, there is one more generalization which we call the {\em Steiner path problem} that asks the following; 
\begin{center}
\minibox[frame]{ 
\emph{\textbf{Instance:}} A connected cograph $G$, a terminal set $X \subseteq V(G)$\\
\emph{\textbf{Question:}} Does there exist a path containing all of $R$ with the least number of vertices from $V(G) \setminus R$.?
}
\end{center}
For a cograph $G$, let $T$ denote its corresponding parse tree.  For a node $v \in T$, let $T_v$ denotes the subtree rooted at $v$ and $G_v$ denotes the cograph corresponding to the cotree $T_v$. For every non-root node $v$ in $T$, we maintain three states $S_v$, $U_v$ and $L_v$, where $S_v$ is the longest path in the graph $G[X \cap V(G_v)]$.  We make use of the paths generated by the recursive subproblems while updating $S_v$. The paths that are not used for updating $S_v$ are included in $U_v$ which may be used for updating the states at the nodes that come later in the postorder traversal of $T$.  Vertices in $X \cap V(G_v)$ but not part of $S_v$ and $U_v$ are included in $L_v$. Finally, when we reach the root node, we add additional vertices to the longest path in the graph $G[X]$ using the states updated at the children of the root node and the path is stored in $S_v$. Let $S_1=(x_1,\ldots ,x_{|S_1|})$, $U_1=(Q_1,Q_2,\ldots, Q_{|U_2|})$ and $S_2=(y_1,\ldots ,y_{|S_2|})$, $U_2=(R_1,R_2,\ldots, R_{|U_2|})$ denote the states w.r.t. the first and second child of $v$, respectively in $T$. Let $r_{jk}$ denote the $k$th vertex in the path $R_j$ in $U_2$. Let $L_2=\{z_1,z_2,\ldots ,z_{|L_2|}\}$. We update the states $S_v$, $U_v$ and $L_v$ as follows.
\begin{enumerate}
\item When $v$ is a leaf node and $v\in X$, $S_v$ contains the vertex $v$ and $U_v$ and $L_v$ are empty. If $v\notin X$, $S_v$ and $U_v$ are empty and $L_v$ contains the vertex $v$.

\item When $v$ is labelled 0, without loss of generality, assume $|S_1|\geq |S_2|$. Now, $S_v$ contains the path $S_1$ and $U_v$ contains the paths in $U_1$, $U_2$ and $S_2$. Let $U_v=(P_1,P_2,\ldots ,P_{|U_v|})$ be the ordering of the paths in $U_v$ such that $|P_1|\geq |P_2|\geq \ldots \geq |P_{|U_v|}|$.

\item When $v$ is labelled 1, without loss of generality, assume $|U_1|\geq |U_2|$. Initialize $S_v$ to $\emptyset$ and assume all paths in $U_1$ and $U_2$ are {\em uncovered} initially. A path in $U_i$ is said to be {\em covered} if it is considered as part of updating $S_v$. Let $a$ and $b$ be the number of paths uncovered in $U_1$ and $U_2$, respectively w.r.t. $S_v$. Let $c$ and $d$ denote the number of vertices uncovered in $S_2$ and $L_2$, respectively w.r.t. $S_v$. If $|U_1|=|U_2|$, then if $|S_1|=0$ and $|S_2|=0$, we extend the path $S_v$ by concatenating a path in $U_2$ and a path $U_1$ alternately (a path in $U_1$ and a path of $U_2$ by preserving the order of paths in $U_1$ and $U_2$) until $a=0$. Otherwise, we extend $S_v$ by concatenating a path in $U_1$ and a path $U_2$ alternately preserving the order of paths in $U_1$ and $U_2$ until $b=0$. Now, $U_v$ is empty and $L_v$ contains the vertices in $L_1$ and $L_2$. If $|U_1|>|U_2|$, we consider the following cases.

\begin{enumerate}
\item If $|S_1|=0$, we extend the path $S_v$ by concatenating a path in $U_1$ and a vertex of a path in $U_2$, that is, the end point of a path in $U_1$ is attached to a vertex of a path in $U_2$ alternately preserving the order of paths in $U_1$ and $U_2$ until $a=b+1$ or $b=0$. We consider the next path in $U_2$ once all the vertices in the path are covered in $S_v$. If $a=b+1$, then we extend $S_v$ by concatenating a path in $U_1$ and a path in $U_2$ preserving the order until $a=0$. Then, we extend $S_v$ by adding $S_2$. Otherwise, if $|S_2|\neq 0$, we extend $S_v$ by concatenating a path uncovered in $U_1$ w.r.t. $S_v$ preserving the order and a vertex in the path $S_2$ alternately until $a=0$ or $c=0$.  If $|S_2|=0$, we extend $S_v$ by adding another uncovered path in $U_1$ w.r.t. $S_v$.

\item Otherwise, if $|S_2|=0$, we extend the path $S_v$ by concatenating a path in $U_1$ and a vertex of a path in $U_2$, that is, the end point of a path in $U_1$ is attached to a vertex of a path in $U_2$ alternately preserving the order of paths in $U_1$ and $U_2$ until $a=b$ or $b<0$. If $a=b$, we concatenate uncovered paths in $U_1$ and $U_2$ w.r.t. $S_v$ alternately until $a=0$ and $b=0$. Further, we extend the path by adding $S_1$ to $S_v$.

\item If $|S_1|\neq 0$ and $|S_2|\neq 0$, we extend the path $S_v$ by concatenating a path in $U_1$ and a vertex of a path in $U_2$, that is, the end point of a path in $U_1$ is attached to a vertex of a path in $U_2$ alternately preserving the order of paths in $U_1$ and $U_2$ until $a=b+1$ or $b<0$. If $a=b+1$, then we extend $S_v$ by concatenating a path in $U_1$ and a path in $U_2$ preserving the order until $a=0$. Finally, we concatenate $S_2$ and then $S_1$ to $S_v$. Otherwise, we extend $S_v$ by concatenating the paths uncovered in $U_1$ w.r.t. $S_v$ and a vertex in the path $S_2$ alternately until $a=0$ or $c=0$.
\item After the above three cases, if $a\neq 0$ and $v$ is a root node, we extend $S_v$ by concatenating the paths uncovered in $U_1$ w.r.t. $S_v$ and a vertex in $L_2$ alternately until $a=0$ or $d=0$. If $|S_1|\neq 0$, we extend the path $S_v$ by concatenating the path $S_1$. Otherwise, we add an uncovered path in $U_1$ to $S_v$. Now, $U_v$ contains the paths uncovered in $U_1$ w.r.t. $S_v$ and $L_v$ contains the vertices in $L_1$ and $L_2$.
\end{enumerate}

\item Finally, when we update the root node, $S_v$ contains the Steiner path in $G$, if it exists.
\end{enumerate}

\subsubsection{The Algorithm}\hfill \\ \\
We shall now present an algorithm for finding the Steiner path in a cograph, if it exists, and further give a proof of correctness for the algorithm.

\begin{algorithm}[hp]
\caption{The Steiner path in a cograph}
\label{spa}
\begin{algorithmic}[1]
\STATE{\textbf{Input:} A cograph $G$ and its parse tree $T$, the terminal set $X \subset V(G)$}
\STATE{\textbf{Output:} The Steiner path $S$ of $G$, if it exists}
\STATE{Perform post order traversal on $T$}
\STATE{\textbf{if} $v$ is a leaf node}
	\STATE{\hspace{0.5cm}If $v\in X$, then $S_v=(v)$, $U_v=\emptyset$ and $L_v=\emptyset$. Otherwise, $S_v=\emptyset$, $U_v=\emptyset$ and $L_v=v$}
\STATE{\textbf{else}}
	\STATE{\hspace{0.5cm}\texttt{/* Let $S_1=(x_1,\ldots ,x_{|S_1|})$, $U_1=(Q_1,Q_2,\ldots, Q_{|U_2|})$ and $S_2=(y_1,\ldots ,y_{|S_2|})$, $U_2=(R_1,R_2,\ldots, R_{|U_2|})$ denote the states w.r.t. the first and second child of $v$, respectively in $T$. Let $r_{jk}$ denote the $k$th vertex in the path $R_j$ in $U_2$. Let $L_2=\{z_1,z_2,\ldots ,z_{|L_2|}\}$. */}}
	\STATE{\hspace{0.5cm}\textbf{if} $v$ is labelled 0}
		\STATE{\hspace{1cm}$S_v\leftarrow S_1$ \texttt{/* Assume $S_1=max(S_1,S_2)$ */}}
		\STATE{\hspace{1cm}$U_v\leftarrow U_1\cup U_2\cup S_2$ and $L_v\leftarrow L_1\cup L_2$. Let $U_v=(S_1,S_2,\ldots ,S_{|U_v|})$ be the ordering of the paths in $U_v$ such that $|S_1|\geq |S_2|\geq \ldots \geq S_{|U_v|}$}
	\STATE{\hspace{0.5cm}\textbf{elseif} $|U_1|=|U_2|$}
		\STATE{\hspace{1cm}\textbf{if} $|S_1|=0$ and $|S_2|\neq 0$ \textbf{then} $S_v=(R_1,Q_1,\ldots ,R_{|U_2|},Q_{|U_1|},S_2)$, $U_v=\emptyset$ and $L_v\leftarrow L_1\cup L_2$}
		\STATE{\hspace{1cm}\textbf{else} $S_v=(Q_1,R_1,\ldots ,Q_{|U_1|},R_{|U_2|},S_1,S_2)$, $U_v=\emptyset$ and $L_v\leftarrow L_1\cup L_2$}
	\STATE{\hspace{0.5cm}\textbf{else}}
		\STATE{\hspace{1cm}Initialize $S_v=\emptyset$, $a=|U_1|$, $b=|U_2|$, $c=|S_2|$, $d=|L_2|$ and $i=j=k=1$ \texttt{/* Assume $|U_1|\geq |U_2|$ */}}
		\STATE{\hspace{1cm}\textbf{if} $|S_1|=0$}
			\STATE{\hspace{1.5cm}\textbf{while} $a>b+1$ and $b>0$}
				\STATE{\hspace{2cm}$S_v=(S_v,Q_i,r_{jk})$; $i=i+1$; $k=k+1$; $a=a-1$}
				\STATE{\hspace{2cm}\textbf{if} $k=|R_j|+1$  \textbf{then} $j=j+1$; $k=1$; $b=b-1$}
			\STATE{\hspace{1.5cm}\textbf{if} $a=b+1$ \textbf{then} $S_v=(S_v,Q_i,R_j,\ldots ,Q_{|U_2|},R_{|U_2|},Q_{|U_1|},S_2)$}
			\STATE{\hspace{1.5cm}\textbf{elseif} $|S_2|\neq 0$}
				\STATE{\hspace{2cm}\textbf{while} $i<=|U_1|$ and $c>0$}
					\STATE{\hspace{2.5cm}$S_v=(S_v,Q_i,y_c)$; $i=i+1$; $c=c-1$}
			\STATE{\hspace{1.5cm}\textbf{else} $S_v=(S_v,Q_i)$}
		\STATE{\hspace{1cm}\textbf{elseif} $|S_2|=0$}
			\STATE{\hspace{1.5cm}\textbf{while} $a>b$ and $b\geq 0$}
				\STATE{\hspace{2cm}$S_v=(S_v,Q_i,r_{jk})$; $i=i+1$; $k=k+1$; $a=a-1$}
				\STATE{\hspace{2cm}\textbf{if} $k=|R_j|+1$ \textbf{then} $j=j+1$; $k=1$; $b=b-1$}
			\STATE{\hspace{1.5cm}\textbf{if} $a=b$  \textbf{then} $S_v=(S_v,Q_i,R_j,\ldots ,Q_{|U_1|},R_{|U_2|},S_1)$}
		\STATE{\hspace{1cm}\textbf{else}}
			\STATE{\hspace{1.5cm}\textbf{while} $a>b+1$ and $b\geq 0$}
				\STATE{\hspace{2cm}$S_v=(S_v,Q_i,r_{jk})$; $i=i+1$; $k=k+1$; $a=a-1$}
				\STATE{\hspace{2cm}\textbf{if} $k=|R_j|+1$  \textbf{then} $j=j+1$; $k=1$; $b=b-1$}
			\STATE{\hspace{1.5cm}\textbf{if} $a=b+1$ \textbf{then} $S_v=(S_v,Q_i,R_j,\ldots ,Q_{|U_2|},R_{|U_2|},Q_{|U_1|},S_2,S_1)$}
			\STATE{\hspace{1.5cm}\textbf{else}}
				\STATE{\hspace{2cm}\textbf{while} $i<=|U_1|$ and $c>0$}
					\STATE{\hspace{2.5cm}$S_v=(S_v,Q_i,y_c)$; $i=i+1$; $c=c-1$}
		\STATE{\hspace{1cm}\textbf{if} $a\neq 0$}
			\STATE{\hspace{1.5cm}\textbf{if} $v=R$}
				\STATE{\hspace{2cm}\textbf{while} $i<=|U_1|$ and $d>0$}
					\STATE{\hspace{2.5cm}$S_v=(S_v,Q_i,z_d)$; $i=i+1$; $d=d-1$}
			\STATE{\hspace{1.5cm}\textbf{if} $|S_1|\neq 0$ \textbf{then} $S_v=(S_v,S_1)$}
			\STATE{\hspace{1.5cm}\textbf{else} $S_v=(S_v,Q_i)$}
		\STATE{\hspace{1cm}$U_v\leftarrow U_1$ and $L_v\leftarrow L_1\cup L_2$}
\STATE{If $|U_v|=0$, output $S_v$, the Steiner path in $G$. Otherwise, print the Steiner path does not exist.}
\end{algorithmic}
\end{algorithm}

\subsubsection{Trace of Algorithm \ref{spa}}\hfill \\ \\
We now trace the steps of Algorithm \ref{spa} in the Figure \ref{fig:sp}.  We traverse through the nodes in the post order traversal of the parse tree. In this example, we traverse in the order $a,b,c,1,d,e,f,2,3,4,5,g,h,6,i,j,k,7,8,9,l,m,10,n,o,11,12,p,q,13,14,15,16$. While traversing, we update $S_v$, $U_v$ and $L_v$ as shown in the figure.  For example, the states at nodes $a,b,c,d,e,f,g,i,j,k,l,m,n,o,p,q$ are updated as per \emph{Step 5} in the algorithm as they are leaf nodes and states at nodes $3,4,5,7,9,10,11,14$ are updated as per \emph{Steps 9-10}. The states at nodes $1,2,6,8,12,13,15,16$ are updated as in \emph{Steps 11-44}.

\begin{figure}[h!]

    \begin{center}
    \includegraphics[height=8cm,width=16cm]{sp.eps}
    \caption{Trace of the Steiner path algorithm when the Steiner path exists}
    \label{fig:sp}
    \end{center}
\end{figure}

\begin{theorem}
Given a cograph $G$ and a set of terminal vertices $X$, the algorithm \ref{spa} outputs the Steiner path in $G$, if it exists.
\end{theorem}
\begin{proof}
To prove the theorem, we show that $S_v$ gives the Steiner path in $G$ if it exists when $v$ is a root node and for every non-root node $v$, $S_v$ gives the longest path in the graph $G[V(G_v)\cap X]$. Further, $U_v$ contains the paths that are not used for updating $S_v$, and $L_v$ contains $V(G_v)\setminus X$. We shall prove the claim by induction on the height of $T_v$, $T$ rooted at vertex $v$.\\
\emph{Basis Step:} When $v$ is a leaf node, the claim is true.\\
\emph{Induction Hypothesis:} Assume that the claim is true for a parse tree $T$ of height $h \leq k$. Let $v$ be a node at height $k+1$, $k \geq 0$.  Let $G_1$ and $G_2$ denote the subgraphs induced by the leaves in the subtrees rooted at the children of the node $v$ in $T$. Let $S_1$ and $S_2$ be the longest paths given by our algorithm in $G[V(G_1)\cap X]$ and $G[V(G_2)\cap X]$, respectively. Let $U_1$ and $U_2$ contains the paths generated by recursive subproblems that are not used for updating $S_1$ and $S_2$, respectively. Let $L_2$ be the set containing the vertices in $V(G_2)\setminus X$. We shall now prove that the claim is true for the node $v$.\\
\emph{Induction Step:} Let $S_1=(x_1,\ldots ,x_{|S_1|})$, $S_2=(y_1,\ldots ,y_{|S_2|})$. Let $U_1=\{Q_1,Q_2,\ldots, Q_{|U_2|}\}$ and $U_2=\{R_1,R_2,\ldots, R_{|U_2|}\}$. Let $r_{jk}$ denote the $k$th vertex in the path and $R_j$. Let $L_2=\{n_1,n_2,\ldots ,n_{|L_2|}\}$. When $v$ is labelled 0, as $G_1$ and $G_2$ are not connected, $max(S_1,S_2)$ (say $S_1$) becomes the longest path in the graph $G[V(G_v)\cap X]$, if $v$ is a non-root node. If $v$ is the root node, then clearly, the Steiner path does not exist. Hence, our algorithm is correct. Further, $U_v$ contains the paths in $U_1$, $U_2$ and $S_2$. All the paths in $U_v$ are sorted in the decreasing order of their lengths. Note that none of the vertices in $U_1$ and $U_2$ are connected to $r_{j1}$ and $r_{j|R_j|}$ because of the node $v$ labelled 0. Also, none of the vertices in $U_1$ and $U_2$ are not connected to $x_1$, $x_{|S_1|}$ and $y_1$, $y_{|S_2|}$, respectively.\\
When $v$ is labelled 1, we do step analysis to prove the claim. Recall that every vertex in $G_1$ is connected to every vertex in $G_2$. Let $a$ and $b$ denote the number of uncovered paths in $P_v$ w.r.t. $U_1$ and $U_2$, respectively. Let $c$ and $d$ denote the number of vertices uncovered in $S_2$ and $L_2$, respectively w.r.t. $S_v$. In \emph{Steps 11-13} of the algorithm, when $|U_1|=|U_2|$, we extend $S_v$ by concatenating a path in $U_2$ and a path in $U_1$ alternately until $a=0$ or vice-versa until $b=0$. This implies that all paths in $U_1$ and $U_2$ are covered. This implies that $S_v$ gives spanning path in $G[V(G_1)\cap X]$ when $v$ is a non-root node. When $v$ is the root node, the Steiner path exists in $G_v$ without adding additional vertices. Thus, our algorithm is correct. If $|U_1|\geq |U_2|$, we have three cases. In the first case, in \emph{Steps 17-18} of the algorithm, we extend $S_v$ by concatenating the path $Q_i$ and the vertex $r_{jk}$ alternately until $a=b+1$ or $b=0$. If $a=b+1$, then we extend $P_v$ by concatenating the path $Q_i$ and the path $R_j$ until $a=0$. Further, we extend the path by adding $S_2$ as in \emph{Step 20}. Therefore, $S_v$ contains the all of $V(G_v)\cap X$. Hence, the path given by our algorithm is correct. Otherwise, if $|S_2|\neq 0$, then extend the path $S_v$ by concatenating the path $Q_i$ and the vertices in $S_2$ until $c=0$ or $a=0$ as in \emph{Steps 22-23}. If $|S_2|=0$, we add an uncovered path in $U_1$ to $S_v$ in \emph{Step 24}. Now, if $a=0$, then it means we have a spanning path which implies our algorithm is correct. In the second case, we first extend the path $S_v$ by concatenating a path in $U_1$ and a vertex of a path in $U_2$ alternately until $a=b$ or $b<0$ as in \emph{Steps 27-28}. If $a=b$, we concatenate uncovered paths in $U_1$ and $U_2$ w.r.t. $S_v$ alternately until $a=0$ and $b=0$. Further, we extend the path by adding $S_1$ to $S_v$ as in \emph{Step 29}. Now, if $a=0$, then we a spanning path in $G[V(G_1)\cap X]$. In the third case, we extend the path $S_v$ by concatenating a path in $U_1$ and a vertex of a path in $U_2$ until $a=b+1$ or $b<0$ as in \emph{Steps 32-33}. If $a=b+1$, then we extend $S_v$ by concatenating a path in $U_1$ and a path in $U_2$ until $a=0$. Finally, we concatenate $S_2$ and then $S_1$ to $S_v$ in \emph{Step 34}. Otherwise, we extend $S_v$ by concatenating the paths uncovered in $U_1$ w.r.t. $S_v$ and a vertex in the path $S_2$ alternately until $a=0$ or $c=0$ in \emph{Steps 36-37}. \\
After these cases, if $v$ is a non-root node, clearly, $S_v$ contains the longest path in the graph $G[V(G_v)\cap X]$ by following the induction argument similar to the proof of Algorithm \ref{lp}. From the induction hypothesis, any path in $G$ has $a$ number of uncovered paths in $U_1$. Suppose $v$ is a root node. If $a=0$, then clearly, $G$ is an yes instance of Steiner path with no additional vertices. If $a\neq 0$ and if $v$ is a root node, we must add minimum $a-1$ additional vertices to cover all paths in $U_1$. Since we must add vertices from $G_2$, we must add minimum $|L_2|$ additional vertices as in \emph{Steps 40-41}. This completes the induction argument.\hfill\(\qed\)
\end{proof}

Given a terminal set $X$ and a graph $G$, Steiner cycle asks for a cycle containing all of $X$ with a minimum number of additional vertices. We shall now give conditions for the existence of Steiner path (cycle) in a cograph, $G$. Let $S_1$, $U_1$ and $S_2$, $U_2$ denote the states w.r.t. the first and second child of $v$, respectively in $T$.

\begin{theorem}
Given a cograph $G$ and a terminal set $X$, there exists a Steiner path in the graph $G$ if and only if $|U_1|\leq |G_2|$.
\end{theorem}
\begin{proof}
\emph{Necessity: }There exists a Steiner path in $G$. This implies we covered all the paths in $U_1$ and $U_2$ while updating $S_v$. On the contrary, assume that $|U_1|=|G_2|+1$. In our algorithm, we concatenate paths in $U_1$ with the vertices in $G_2$ alternately to the Steiner path. We further extend the path till $S_1$. Then, we have one path uncovered in $S_v$ w.r.t $U_1$ which is a contradiction to the definition of Steiner path.\\
\emph{Sufficiency: }$|U_1| \leq |G_2|$. We concatenate paths in $U_1$ with the vertices in $G_2$ alternately to $S_v$. Thus, we cover all paths in $U_1$. Finally, we concatenate $S_1$ to the path. Therefore, it follows that the graph $G_v$ has a Steiner path. \qed
\end{proof}

\begin{theorem}
Given a cograph $G$ and a terminal set $X$, there exists a Steiner cycle in the graph $G$ if and only if $|U_1|\leq |G_2|-1$.
\end{theorem}
\begin{proof}
\emph{Necessity: }There exists a Steiner cycle in $G_v$. This implies we covered all the paths in $U_1$ while updating the path $S_v$ and last vertex we concat to the path must be from $G_2$ as none of the paths from $U_1$ and $S_1$ are connected. On the contrary, assume that $|U_1|=|G_2|$. In our algorithm, we concat paths in $U_1$ with the vertices in $G_2$ alternately to the Steiner path. Finally, we concatenate $S_1$ to the path. No vertex in $U_1$ is connected to $S_1$, which is a contradiction to the definition of Steiner cycle.\\
\emph{Sufficiency: }$|U_1|\leq |G_2|-1$. We concatenate paths in $U_1$ with the vertices in $G_2$ alternately to $S_v$. Then, we concat $S_1$ to the path. Finally, we concatenate the uncovered vertex in $S_v$ w.r.t. $G_2$ to the path. Since the first vertex in the path is from $G_1$ and the last vertex is from $G_2$, it follows that $G$ has Steiner cycle.\qed
\end{proof}

\subsection{The Minimum Leaf Spanning Tree problem}
We next consider another optimization problem, namely the minimum leaf spanning tree problem in cographs which is also a generalization of the Hamiltonian path problem.  In graph classes where the Hamiltonian path problem is polynomial-time solvable, it is natural to study the complexity of the minimum leaf spanning tree problem, which is defined as follows;
\begin{center}
\minibox[frame]{
\emph{\textbf{Instance:}} A connected cograph $G$\\
\emph{\textbf{Question:}} Does there exist a spanning tree $H$ with a minimum number of leaves in $G$
}
\end{center}
\subsubsection{Outline of the Algorithm}
Let $T_1$ and $T_2$ denote the subtrees rooted at the children of the root node in $T$. Our algorithm first runs the longest path algorithm on the subtrees $T_1$ and $T_2$. Let $P_1=(x_1,\ldots ,x_{|P_1|})$and $P_2=(y_1,\ldots ,y_{|P_2|})$ denote the states w.r.t. $T_1$ and $T_2$, respectively. Let $U_1=\{Q_1,Q_2,\ldots, Q_{|U_2|}\}$ be the state of $T_1$. Let $U_2=\{R_1,R_2,\ldots, R_{|U_2|}\}$ denote the state of $T_2$ and $r_{ij}$ denote the $j$th vertex in the path $R_i$. Without loss of generality, assume $|U_1|\geq |U_2|$. Initialize $P$ to $\emptyset$. Let $a$ and $b$ denote the number of paths uncovered in $U_1$ and $U_2$, respectively w.r.t. $P$. Let $c$ denote the number of vertices uncovered in $P$ w.r.t. $P_2$. Firstly, we extend the path $P$ by concatenating a path in $U_1$ and a vertex of a path in $U_2$ alternately by preserving the order until $a=b+1$ or $b=0$. If $a=b+1$, then we extend $P$ by concatenating the paths in $U_1$ and $U_2$ alternately preserving the order until $a=0$. Further, $H$ contains the path $P$. Otherwise, we extend the path $P$ by concatenating the paths in $U_1$ and vertices in $P_2$ alternately until $a=0$ or $c=0$. If $a=0$, then $H$ contains the path $P$. Otherwise, let $P=(v_1,v_2,\ldots ,v_k)$ and initialize $H=P$. We then concatenate $P_1$ and all the uncovered paths in $U_1$ to the $v_k$. After this procedure, $H$ contains a minimum leaf spanning tree of the cograph $G$.
\subsubsection{The Algorithm}\hfill \\ \\
We now present an algorithm for finding a minimum leaf spanning tree in a cograph and further give a proof of correctness of the algorithm.
\begin{algorithm}[H]
\caption{Minimum leaf spanning tree in a cograph}
\label{mlst}
\begin{algorithmic}[1]
\STATE{\textbf{Input:} A connected cograph $G$ and a parse tree $T$}
\STATE{\textbf{Output:} A minimum leaf spanning tree $H$}
\STATE{Compute the longest path in $T_1$ and $T_2$}
\STATE{\texttt{/* Let $P_1=(x_1,\ldots ,x_{|P_1|})$, $P_2=(y_1,\ldots ,y_{|P_2|})$. Let $U_1=\{Q_1,Q_2,\ldots, Q_{|U_2|}\}$ and  $U_2=\{R_1,R_2,\ldots, R_{|U_2|}\}$. Let $r_{jk}$ denote the $k$th vertex in the path $R_j$ in $U_2$*/}}
\STATE{Initialize $P=\emptyset$, $a=|U_1|$, $b=|U_2|$, $c=|P_2|$ and $i=j=k=1$ \texttt{/* Assume $|U_1|\geq |U_2|$ */}}
\STATE{\textbf{while} $a>b+1$ and $b>0$}
		\STATE{\hspace{0.5 cm}$P=(P,Q_i,r_{jk})$; $i=i+1$; $k=k+1$; $a=a-1$}
		\STATE{\hspace{0.5cm}\textbf{if} $k=|R_j|+1$ \textbf{then} $j=j+1$; $k=1$; $b=b-1$}
		\STATE{\textbf{if} $a=b+1$ \textbf{then} $P=(P,Q_i,R_j,Q_{i+1},R_{j+1},\ldots ,Q_{|U_2|},R_{|U_2|},Q_{|U_1|},P_2,P_1)$ and $H=P$}
			\STATE{\textbf{else}}
			\STATE{\hspace{0.5cm}\textbf{while} $i<=|U_1|$ and $c>0$}
		\STATE{\hspace{1cm}$P=(P,Q_i,y_c)$; $i=i+1$; $c=c-1$}
		\STATE{\hspace{0.5cm}\textbf{if} $i>|U_1|$ \textbf{then} $H=(P,P_1,P_2)$}
		\STATE{\hspace{0.5cm}\textbf{else}}
		\STATE{\hspace{1cm}Let $P=(v_1,v_2,\ldots ,v_k)$ and initialize $H=P$}
		\STATE{\hspace{1cm}Add the edges $\{v_k,q_{l1}\}$ and $E(Q_l)\forall i\leq l\leq |U_1|$ to $H$}
		\STATE{\hspace{1cm}Add the edges $\{v_k,x_1\}$ and $E(P_1)$ to $H$}
\STATE{Output $H$, a minimum leaf spanning tree in $G$}
\end{algorithmic}
\end{algorithm}

\subsubsection{Trace of the algorithm: (Algorithm \ref{mlst})}
We now trace the steps of Algorithm \ref{mlst} in Figure \ref{fig:mlst}. In this example, we compute the longest paths in $T_1$ and $T_2$ as in \emph{Step 3}. After updating the states from the algorithm, we find a minimum leaf spanning tree following the \emph{Steps 5-15}.

\begin{figure}[h!]
    \begin{center}
    \includegraphics[height=13cm,width=17cm]{mlst.eps}
    \caption{Trace of the minimum leaf spanning tree algorithm}
    \label{fig:mlst}
    \end{center}
\end{figure}
\subsubsection{Proof of correctness of Algorithm \ref{mlst}}
Our algorithm runs longest path algorithm on the two subtrees rooted at the children of the root node in $T$ in \emph{Step 3}. Using the states updated at the two children, we find the longest path possible in $G$ as in \emph{Steps 5-12}. If there are no uncovered paths in $U_1$, then it implies there is an yes instance of Hamiltonian path in $G$. Therefore, $H$ contains the path $P$. Otherwise, we concatenate all uncovered paths in $U_1$ and the path $P_1$ to the end vertex of the path $P$ in \emph{Steps 15-17}. Recall that the end vertex of $P$ is from $G_2$ and it is connected to all vertices in $G_1$. And none of the paths in $U_1$ and $S_1$ are connected. Thus, $H$ contains a minimum leaf spanning tree in $G$. \\ \\
{\bf Conclusions and Directions for further research:} In this paper, we have initiated the structural understanding of cographs from the perspective of minimal vertex separators.  Further, using the structural results we have enumerated all minimal vertex separators and presented polynomial-time algorithms for some connectivity augmentation problems. Subsequently, we looked at three classical problems such as Hamiltonian path (cycle), Steiner path and minimum leaf spanning tree in cographs, and presented polynomial-time algorithms for all of them.  In the context of edge connectivity augmentation, we presented polynomial algorithms without preserving cograph property. We believe that these results can be extended to preserve the cograph property.  We also believe that complexity of domination and its variants in cographs can be found using dynamic programming on the underlying parse tree.


\newpage
\begin{thebibliography}{4}
\bibitem{west}
D. B. West: Introduction to graph theory. Vol.2, Prentice hall, (2001).

\bibitem{golumbic}
M. C. Golumbic: Algorithmic graph theory and perfect graphs. Vol. 57, Elsevier, (2004).

\bibitem{kirkpatrick}
David G Kirkpatrick and T Przytycka: Parallel recognition of complement reducible graphs and cotree construction. Discrete Applied Mathematics, Vol.29, No.1, pp.79-96, (1990).

\bibitem{corneil}
Derek G Corneil, Yehoshua Perl, and Lorna K Stewart: A linear recognition algorithm for cographs. SIAM Journal
on Computing, Vol.14, No.4, pp.926-934, (1985).

\bibitem{lerchs}
Derek G Corneil, H Lerchs, and L Stewart Burlingham: Complement reducible graphs. Discrete Applied Mathematics, Vol.3, No.3, pp.163-174, (1981).

\bibitem{jacobus}
Antonius Jacobus Johannes Kloks: Minimum fill-in for chordal bipartite graphs. Vol.93, (1993).

\bibitem{narayanaswamy}
N. S. Narayanaswamy and N. Sadagopan: Connected $(s,t)$-Vertex Separator Parameterized by Chordality. Journal of Graph Algorithms and Applications, Vol.19, pp.549-565, (2015).

\bibitem{manogna}
S.Dhanalakshmi, N.Sadagopan and V.Manogna: On $2K_2$ -free graphs. International Journal of Pure and Applied Mathematics, Vol.109, No.7, pp.167-173, (2016).

\bibitem{dragan}
A.Brandstadt, F.F.Dragan, V.B.Le and T.Szymczak: On stable cutsets in graphs. Discrete Applied Mathematics, Vol.105, pp.39-50, (2000).

\bibitem{tarjan}
R.E. Tarjan: Decomposition by clique separators. Discrete Math, Vol.55, pp.221-232, (1985).

\bibitem{gabow}
Gabow and Harold N: Data Structures for Weighted Matching and Extensions to $b$-matching and $f$-factors. arxiv.org, doi: \url{https://arxiv.org/abs/1611.07541}

\bibitem{yu}
Yu, Ming-Shing and Yang, Cheng-Hsing: An $O(n)$ time algorithm for maximum matching on cographs. Information processing letters, Vol.47, No.2, pp.89-93, (1993)


\bibitem{andras}
A. Brandst{\"a}dt, Van Bang Le and J. P. Spinrad: Graph classes: a survey. SIAM, (1999).

\bibitem{corneilrecog}
Bretscher, Anna and Corneil, Derek and Habib, Michel and Paul, Christophe: A simple linear time LexBFS cograph recognition algorithm. International Workshop on Graph-Theoretic Concepts in Computer Science, pp. 119-130, (2003).

\bibitem{shew}
S. Shew: A Cograph Approach to Examination Scheduling. Master’s thesis, Department of Computer Science, University of Toronto, Toronto, Canada, (1986).

\bibitem{nikolopoulos}
Stavros D Nikolopoulos and Leonidas Palios: Minimal separators in $P_4$-sparse graphs. Discrete mathematics, Vol.306, No.3, pp.381–392, (2006).

\bibitem{eshwaran}
Kapali P Eswaran and R Endre Tarjan: Augmentation problems. SIAM Journal on Computing, Vol.5, No.4, pp.653–665, (1976).


\bibitem{ioannidou}
Ioannidou, Kyriaki and Mertzios, George B and Nikolopoulos, Stavros D: The longest path problem has a polynomial solution on interval graphs. Algorithmica, Vol.61, No.2, pp.320-341, (2011).

\bibitem{mertzios}
Mertzios, George B and Corneil, Derek G: A simple polynomial algorithm for the longest path problem on cocomparability graphs. SIAM Journal on Discrete Mathematics, Vol.26, No.3, pp.940-963, (2012).

\bibitem{uehara}
Uehara, Ryuhei and Valiente, Gabriel: Linear structure of bipartite permutation graphs and the longest path problem. Information Processing Letters, Vol.103, No.2, pp.71-77, (2007).

\bibitem{karger}
Karger, David and Motwani, Rajeev and Ramkumar, GDS: On approximating the longest path in a graph. Algorithmica, Vol.18, No.1, pp.82-98, (1997).

\bibitem{jung}
Jung, Heinz A: On a class of posets and the corresponding comparability graphs. Journal of Combinatorial Theory, Series B, Vol.24, No.2, pp.125-133, (1978).

\bibitem{watanabe}
Toshimasa Watanabe and Akira Nakamura: Edge-connectivity augmentation problems. Journal of Computer and System Sciences, Vol.35, No.1, pp.96–144, (1987).

\bibitem{frank}
Andr{\"a}s Frank: Augmenting graphs to meet edge-connectivity requirements. SIAM Journal on Discrete Mathematics, Vol.5, No.1, pp.25–53, (1992).



\bibitem{kloks}
Kloks, Ton and Kratsch, Dieter: Listing all minimal separators of a graph. SIAM Journal on Computing, Vol.27, No.3, pp. 605-613, (1998).


\bibitem{arnborg}
Arnborg, Stefan and Proskurowski, Andrzej: Linear time algorithms for NP-hard problems restricted to partial k-trees. Discrete applied mathematics, Vol.23, No.1, pp.(1989).

\bibitem{vegh}
L{\"a}szl{\"o} A {\"V}egh: Augmenting undirected node-connectivity by one. SIAM Journal on Discrete Mathematics, Vol.25, No.2, pp.695-718, (2011).

\bibitem{jordan}
Andr{\"a}s Frank and Tibor Jord{\"a}n: Minimal edge-coverings of pairs of sets. Journal of Combinatorial Theory, Series B, Vol.65, No.1, pp.73-110, (1995).

\bibitem{nsn}
Narayanaswamy, NS and Sadagopan, N: A Unified Framework For Bi (Tri) connectivity and Chordal Augmentation. International Journal of Foundations of Computer Science, Vol.21, No.1, pp.67-93, (2013).

\bibitem{hans}
Bodlaender, Hans L: Dynamic programming on graphs with bounded treewidth. International Colloquium on Automata, Languages, and Programming, pp.105-118, (1988).


\bibitem{chang}
Chang, Maw-Shang and Hsieh, Sun-Yuan and Chen, Gen-Huey: Dynamic programming on distance-hereditary graphs. International Symposium on Algorithms and Computation, pp.344-353, (1997).

\bibitem{martin}
Gr{\"o}tschel, Martin and Monma, Clyde L and Stoer, Mechthild: Design of survivable networks. Handbooks in operations research and management science, Vol.7, pp.617-672, (1995).



\end{thebibliography}
\end{document}